\documentclass{article}

\PassOptionsToPackage{numbers, compress}{natbib}

\usepackage[preprint]{nips_2018}

\usepackage[utf8]{inputenc} 
\usepackage[T1]{fontenc}    
\usepackage{hyperref}       
\usepackage{url}            
\usepackage{booktabs}       
\usepackage{amsfonts}       
\usepackage{nicefrac}       
\usepackage{microtype}      
\usepackage[linesnumbered,ruled,noend]{algorithm2e}

\usepackage{xcolor}
\usepackage{gensymb} 
\usepackage{graphicx}
\usepackage{amsmath}
\usepackage{dsfont}
\usepackage{amsthm}
\usepackage{multirow} 
\usepackage{amssymb} 

\usepackage{apptools}
\AtAppendix{\counterwithin{lemma}{section}}
\AtAppendix{\counterwithin{theorem}{section}}
\AtAppendix{\counterwithin{proposition}{section}}
\AtAppendix{\counterwithin{claim}{section}}

\newcount\Comments  
\newcount\Includeappendix 
\newcount\Includeackerc 

\definecolor{darkgreen}{rgb}{0,0.6,0}
\newcommand{\kibitz}[2]{\ifnum\Comments=1{\color{#1}{#2}}\fi}

\newcommand{\omer}[1]{\kibitz{blue}{[Omer :#1]}}

\usepackage{MnSymbol} 

\makeatletter
\newcommand{\RemoveAlgoNumber}{\renewcommand{\fnum@algocf}{\AlCapSty{\AlCapFnt\algorithmcfname}}}
\newcommand{\RevertAlgoNumber}{\algocf@resetfnum}
\makeatother
\RemoveAlgoNumber

\newcommand{\mX}{\mathcal{X}}

\newcommand{\mI}{\mathcal{I}}
\newcommand{\mF}{\mathcal{F}}
\newcommand{\mY}{\mathcal{Y}}
\newcommand{\mZ}{\mathcal{Z}}
\newcommand{\mH}{\mathcal{H}}

\newcommand{\mG}{\mathcal{G}}
\newcommand{\mD}{\mathcal{D}}
\newcommand{\mR}{\mathcal{R}}

\newcommand{\mN}{\mathcal{N}}
\newcommand{\mS}{\mathcal{S}}
\newcommand{\mT}{\mathcal{T}}
\newcommand{\ind}{\mathds{1}}

\newcommand{\R}{\mathbb{R}}

\newcommand{\defeq}{\stackrel{\text{def}}{=}}
\newcommand{\vc}{\textnormal{VCdim}}
\newcommand{\pdim}{\textnormal{Pdim}}

\newcommand{\Var}{\mathrm{Var}}

\theoremstyle{plain}
\newtheorem{claim}{Claim}

\newtheorem{theorem}{Theorem}

\newtheorem{lemma}{Lemma}
\newtheorem{corollary}{Corollary}

\theoremstyle{definition}
\newtheorem{example}{Example}

\newcommand\bl[1]{\boldsymbol{ #1 } }
\newcommand{\norm}[1]{\left\lVert#1\right\rVert}
\newcommand\abs[1]{\left| #1  \right|}
\DeclareMathOperator*{\argmax}{arg\,max} 

\DeclareMathOperator{\E}{\mathbb{E}}

\DeclareMathOperator{\sign}{sign}
\title{Competing Prediction Algorithms}

\author{
  Omer Ben{-}Porat \\
  Technion - Israel Institute of Technology\\
  Haifa 32000 Israel\\
  \texttt{omerbp@campus.technion.ac.il} 
  \And 
  Moshe Tennenholtz\\
  Technion - Israel Institute of Technology\\
  Haifa 32000 Israel\\
  \texttt{moshet@ie.technion.ac.il} \\
}

\begin{document}

\maketitle
{\centering \large \textbf{Remark}\vspace{5mm}\\} An updated and significantly improved version of this paper was published in Economics and Computation 2019 under the name ``Regression Equilibrium'', and is publicly available here: \url{https://arxiv.org/abs/1905.02576}. Please refer to that version.\\ 

\begin{abstract}

Prediction is a well-studied machine learning task, and prediction algorithms are core ingredients in online products and services.
Despite their centrality in the competition between online companies who offer prediction-based products, the \textit{strategic} use of prediction algorithms remains unexplored. The goal of this paper is to examine strategic use of prediction algorithms.
We introduce a novel game-theoretic setting that is based on the PAC learning framework, where each player (aka a prediction algorithm at competition) seeks to maximize the sum of points for which it produces an accurate prediction and the others do not. We show that algorithms aiming at generalization may wittingly miss-predict some points to perform better than others on expectation. We analyze the empirical game, i.e. the game induced on a given sample, prove that it always possesses a pure Nash equilibrium, and show that every better-response learning process converges. Moreover, our learning-theoretic analysis suggests that players can, with high probability, learn an approximate pure Nash equilibrium for the whole population using a small number of samples. 
\end{abstract}

\section{Introduction}

Prediction plays an important role in twenty-first century economics.  An important example is the way online retailers advertise services and products tailored to predict individual taste.  Companies collect massive amounts of data and employ sophisticated machine learning algorithms to discover patterns and seek connections between different user groups.  A company can offer customized products, relying on user properties and past interactions, to outperform the one-size-fits-all approach. For instance, after examining sufficient number of users and the articles they read, media websites promote future articles predicted as having a high probability of satisfying a particular user.

For revenue-seeking companies, prediction is another tool that can be exploited to increase revenue. When companies' products are alike, the chance that a user will select the product of a particular company decreases.  In this case a company may purposely avoid offering the user this product and offer an alternative one in order to maximize the chances of having its product selected.
Despite the intuitive clarity of the tradeoff above and the enormous amount of work done on prediction in the machine learning and statistical learning communities, far too little attention has been paid to the study of prediction in the context of \textit{competition}. 

In this paper we introduce what is, to the best of our knowledge, a first-ever attempt to study how the selection of prediction algorithms is affected by strategic behavior in a competitive setting, using a game-theoretic lens. We consider a space of users, where each user is modeled as a triplet $(x,y,t)$ of an instance, a label and a threshold, respectively. A user's instance is a real vector that encodes his\footnote{For ease of exposition, third-person singular pronouns are ``he'' for a user and ``she'' for a player.} properties; the label is associated with his taste, and the threshold is the ``distance'' he is willing to accept between a proposed product and his taste. Namely, the user associated with $(x,y,t)$ embraces a customized product $f(x)$ if $f(x) - y$ is less than or equal to $t$. In such a case, the user is \textit{satisfied} and willing to adopt the product. If a user is satisfied with several products (of several companies), he selects one uniformly at random. Indeed, the user-model we adopt is aligned with the celebrated ``Satisficing" principle of \citeauthor{simon1956rational} \cite{simon1956rational}, and other widely-accepted models in the literature on choice prediction, e.g. the model of selection based on small samples \cite{barron2003small,erev2010choice}.
Several players are equipped with infinite strategy spaces, or hypothesis classes in learning-theoretic terminology. A player's strategy space models the possible predictive functions she can employ. Players are competing for the users, and a player's payoff is the expected number of users who select her offer. To model uncertainty w.r.t. the users' taste, we use the PAC-learning framework of \citeauthor{valiant1984theory} \cite{valiant1984theory}. We assume user distribution is unknown, but the players have access to a sequence of examples, containing instances, labels and thresholds, with which they should optimize their payoffs w.r.t. the unknown underlying user distribution.

From a machine learning perspective we now face the challenge of what would be a good prediction algorithm profile, i.e. a set of algorithms for the players such that no player would deviate from her algorithm assuming the others all stick to their algorithms. Indeed, such a profile of algorithms determines a \textit{pure Nash equilibrium} (PNE) of prediction algorithms, a powerful solution concept which rarely exists in games. An important question in this regard is whether such a profile exists. An accompanying question is whether a learning dynamics in which players may change their prediction algorithms to better-respond to others would converge. 
Therefore, we ask:

$\bullet$ Does a PNE exist?\\
$\bullet$ Will the players be able to find it efficiently with high probability using a better-response dynamics?

We prove that the answer to both questions is yes. 
We first show that when the capacity of each strategy space is bounded (i.e., finite pseudo-dimension), players can learn payoffs from samples. Namely, we show that the payoff function of each player uniformly converges over all possible strategy profiles (that include strategies of the other players). Thus with high probability a player's payoff under any strategy profile is not too distant from her empirical payoff. 
Later, we show that an empirical PNE always exists, i.e., a PNE of the game induced on the empirical sample distribution. Moreover, we show that any learning dynamics in which players improve their payoff by more than a non-negligible quantity converges fast to an approximate PNE. 
Using the two latter results, we show an interesting property of the setting: the elementary idea of sampling and better-responding according to the empirical distribution until convergence leads to an approximate PNE of the game on the whole population. We analyze this learning process, and formalize the above intuition via an algorithm that runs in polynomial time in the instance parameters, and returns an approximate PNE with high probability.
Finally, we discuss the case of infinite capacities, and demonstrate that non-learnability can occur even if the user distribution is known to all players. 

\paragraph{Related work}
The intersection of game theory and machine learning has increased rapidly in recent years. Sample empowered mechanism design \cite{nisan1999algorithmic} is a fruitful line of research. For example, \cite{cole2014sample,GonczarowskiN17,morgenstern2015pseudo} reconsider  auctions where the auctioneer can sample from bidder valuation functions, thereby relaxing the assumption of prior knowledge on bidder valuation distribution \cite{myerson1981optimal}.
Empirical distributions also play a key role in other lines of research  \cite{althofer1994sparse,babichenko2016empirical,lipton2003playing}, where e.g.  \cite{babichenko2016empirical} show how to obtain an approximate equilibrium by sampling any mixed equilibrium. The PAC-learning framework proposed by \citeauthor{valiant1984theory} \cite{valiant1984theory} has also been extended by \citeauthor{blum2017collaborative} \cite{blum2017collaborative}, who consider a collaborative game where players attempt to learn the same underlying prediction function, but each player has her own distribution over the space. In their work each player can sample from her own distribution, and the goal is to use information sharing among the players to reduce the sample complexity.

Our work is inspired by Dueling Algorithms \cite{immorlica2011dueling}. \citeauthor{immorlica2011dueling} analyze an optimization problem from the perspective of competition, rather than from the point of view of a single optimizer. Our model is also related to Competing Bandits \cite{MansourSW18}.
\citeauthor{MansourSW18} consider a competition between two bandit algorithms faced
with the same sample, where users arrive one by one and choose between the two algorithms.
In our work players also share the same sample, but we consider an offline setting and not an online one;  infinite strategy spaces and not a finite set of actions; context in the form of property vector for each user; and an arbitrary number of asymmetric players, where asymmetry is reflected in the strategy space of each player.

Most relevant to our work is \cite{porat2017best}. The authors present a learning task where a newcomer agent is given a sequence of examples, and wishes to learn a best-response to the players already on the market. They assume that the agent can sample triplets composed of instance, label and current market prediction, and define the agent's payoff as the proportion of points (associated with users) she predicts better than the other players. Indeed, \cite{porat2017best} introduces a learning task incorporating economic interpretation into the objective function of the (single) optimizer, but in fact does not provide any game-theoretic analysis. In comparison, this paper considers game-theoretic interaction between players, and its main contribution lies in the analysis of such interactions. 
Since learning dynamics consists of steps of unilateral deviations that improve the deviating player's payoff, the Best Response Regression of \citeauthor{porat2017best}\cite{porat2017best} can be thought of as an initial step to this work.

\paragraph{Our contribution}
Our contribution is three-fold. First, we explicitly suggest that prediction algorithms, like other products on the market, are in competition.
This novel view emphasizes the need for stability in prediction-based competition similar to  \citeauthor{Hotelling}'s stability in spatial competition \cite{Hotelling}.

Second, we introduce an extension of the PAC-learning framework for dealing with strategy profiles, each of which is a sequence of functions. We show a reduction from payoff maximization to loss minimization, which is later used to achieve bounds on the sample complexity for uniform convergence over the set of profiles. We also show that when players have approximate better-response oracles, they can learn an approximate PNE of the empirical game. The main technical contribution of this paper is an algorithm which, given $\epsilon,\delta$, samples a polynomial number of points in the game instance parameters, runs any $\epsilon$-better-response dynamics, and returns an $\epsilon$-PNE with probability of at least $1-\delta$.

Third, we consider games with at least one player with infinite pseudo-dimension. We show a game instance where each player can learn the best prediction function from her hypothesis class if she were alone in the game, but a PNE of the empirical game is not generalized. This inability to learn emphasizes that strategic behavior can introduce further challenges to the machine learning community.

\section{Problem definition}
In this section we formalize the model. We begin with an informal introduction to elementary concepts in both game theory and learning theory that are used throughout the paper.

\paragraph{Game theory}
A non-cooperative game is composed of a set of players $\mN=\{1,\dots N\}$; a strategy space $\mH_i$ for every player $i$; and a payoff function $\pi_i:\mH_1 \times \cdots \times \mH_N \rightarrow \mathbb R$ for every player $i$. The set $\mH =\mH_1 \times \cdots \times \mH_N$ contains of all possible strategies, and a tuple of strategies $\bl h =(h_1,\dots h_N) \in \mH$ is called a \textit{strategy profile}, or simply a profile. We denote by $\bl h_{-i}$ the vector obtained by omitting the $i$-th component of $\bl h$.

A strategy $h_i'\in \mH_i $ is called a \textit{better response} of player $i$ with respect to a strategy profile $\bl h$ if $\pi_i(h_i', \bl h_{-i}) > \pi_i(\bl h)$. Similarly, $h_i'$ is said to be an \textit{$\epsilon$-better response} of player $i$ w.r.t. a strategy profile $\bl h$ if $\pi_i(h_i', \bl h_{-i}) \geq  \pi_i(\bl h) +\epsilon$, and a \textit{best response} to $\bl h_{-i}$ if $\pi_i(h_i', \bl h_{-i}) \geq \sup_{h_i\in \mH_i}\pi_i(h_i,\bl h_{-i})$ .

We say that a strategy profile $\bl h$ is a \textit{pure Nash equilibrium} (herein denoted PNE) if every player plays a best response under $\bl h$. We say that a strategy profile $\bl h$ is an $\epsilon$-PNE if no player has an $\epsilon$-better response under $\bl h$, i.e. for every player $i$ it holds that $\pi_i(\bl h) \geq \sup_{h_i'\in \mH_i}\pi_i(h_i',\bl h_{-i})-\epsilon$.

\paragraph{Learning theory}

Let $F$ be a class of binary-valued functions $F\subseteq {\{0,1\}}^\mX$. Given a sequence $\mS=(x_1,\dots x_m)\in \mX^m$, we denote the \textit{restriction} of $F$ to $\mS$ by $F\cap \mS= \left\{ \left(f(x_1),\dots,f(x_m)\right)\mid f\in F  \right\}$. The \textit{growth function} of $F$, denoted $\Pi_F:\mathbb N \rightarrow \mathbb N$, is defined as $\Pi_F(m) = \max_{\mS\in \mX^m}\abs{F\cap \mS}$. We say that $F$ \textit{shatters} $\mS$ if $\abs{F\cap \mS}=2^{\abs{\mS}}$. The Vapnik-Chervonenkis dimension of a binary function class is the cardinality of the largest set of points in $\mX$ that can be shattered by $F$, 
$\vc(F)= \max\left\{m\in \mathbb N :\Pi_F(m)=2^m    \right\}$.

Let $H$ be a class of real-valued functions $H\subseteq \mathbb R^\mX$. The restriction of $H$ to $\mS \in \mX^m$ is analogously defined, $H\cap \mS= \left\{ \left(h(x_1),\dots,h(x_m)\right)\mid h\in H  \right\}$. We say that $H$ \textit{pseudo-shatters} $\mS$ if there exists $\bl r=(r_1,\dots,r_m)\in \mathbb R^m$ such that for every binary vector $\bl b=(b_1,\dots b_m)\in \{-1,1  \}^m$ there exists $h_{\bl b}\in H$ and for every $i\in [m]$ it holds that $\sign(h_{\bl b}(x_i)-r_i)=b_i $. The \textit{pseudo-dimension} of $H$ is the cardinality of the largest set of points in $\mX$ that can be pseudo-shattered by $H$, 
\[
\pdim(H)=\max\left\{m\in \mathbb N :\exists \mS\in\mX^m \text{ such that } \mS \text{ is pseudo-shattered by }  H \right\}.
\]

\subsection{Model}
\label{subsec:model}
We consider a set of users who are interested in a product provided by a set of competing players. Each user is associated with a vector $(x,y,t)$, where $x$ is the instance; $y$ is the label; and $t$ is the threshold that the user is willing to accept.

The players offer customized products to the users. When a user associated with a vector $(x,y,t)$ approaches player $i$, she produces a prediction $h_i(x)$. If $\abs{h_i(x)-y}$ is at most $t$, the user associated with $(x,y,t)$ will grant one monetary unit to player $i$. Alternatively, that user will move on to another player. We assume that users approach players according to the uniform distribution, although our model and results support any distribution over player orderings. Player $i$ has a set of possible strategies (prediction algorithms) $\mH_i$, from which she has to decide which one to use. Each player aims to maximize her expected payoff, and will act strategically to do so.

Formally, the game is a tuple $\langle \mZ,\mD,\mN,(\mH_i)_{i\in \mN} \rangle$ such that
\begin{enumerate}
\item $\mZ$ is the examples domain $\mZ = \mX \times \mY \times \mT$, where $\mX\subset \R^n$ is the instance domain;  $\mY \subset \R$ is the label domain; and $\mT \subset \R_{\geq 0}$ is the tolerance domain.
\item $\mD$ is a probability distribution over $\mZ = \mX \times \mY \times \mT$.
\item $\mN$ is the set of players, with $\abs{\mN} = N$. A strategy of player $i$ is an element from $\mH_i \subseteq \mY ^ \mX$. The space of all strategy profiles is denoted by $\mH = \bigtimes_{i=1}^N \mH_i  $.

\item For $z=(x,y,t)$ and a function $g:\mX\rightarrow \mY$, we define the indicator $\mI(z,g)$ to be 1 if the distance between the value $g$ predicted for $x$ and $y$ is at most $t$. Formally,
\[
\mI(z,g)=
\begin{cases}
1 & \abs{g(x)-y} \leq t \\
0 & \text{otherwise}
\end{cases}.
\]
\item Given a strategy profile $\bl h=(h_1,\dots h_N)$ with $h_i \in \mH_i$ for $i\in\{1,\dots N\}$ and $z=(x,y,t)\in \mZ$, let
\[
w_i(z;\bl h)=
\begin{cases}
0 & \text{ if } \mI(z,h_i)=0\\
\frac{1}{\sum_{i'=1}^N \mI(z,h_{i'})} & \text{otherwise}
\end{cases}.
\]
Note that $w_i(z;\bl h)$ represents the expected payoff of player $i$ w.r.t. the user associated with $z$. The payoff of player $i$ under $\bl h$ is the average sum over all users, and is defined by
\[
\pi_i (\bl h)= \E_{z\sim \mD}\left[ w_i(z;\bl h) \right].
\]
\item $\mD$ is unknown to the players.
\end{enumerate}

We assume players have access to a sequence of examples $\mS$. Given a game instance $\langle \mZ,\mD,\mN,(\mH_i)_{i\in \mN} \rangle$ and a sample $\mS=\{z_1,\dots z_m\}$, we denote by $\langle \mZ,\mS \sim \mD^m ,\mN,(\mH_i)_{i\in \mN} \rangle$ the \textit{empirical game}: the game over the same $\mN,\mH,\mZ$ and uniform distribution over the known $\mS \in \mZ^m$. We denote the payoff of player $i$ in the empirical game by 
\[
\pi_i^\mS (\bl h)= \E_{z \in \mS}\left[ w_i(z;\bl h) \right]=\frac{1}{m}\sum_{j=1}^m w_i(z_j;\bl h). 
\]
When $\mS$ is known from the context, we occasionally use the term \textit{empirical} PNE to denote a PNE of the empirical game. Since the empirical game is a complete information game, players can use the sample in order to optimize their payoffs. 

The optimization problem of finding a best response in our model is intriguing in its own right and deserves future study.
In this paper, we assume that each player $i$ has a polynomial $\epsilon$-better-response oracle. Namely, given a real number $\epsilon>0$, a strategy profile $\bl h$ and sample $\mS$, we assume that each player $i$ has an oracle that returns an $\epsilon$-better response to $\bl h_{-i}$ if such exists or answers false otherwise, which runs in time $\text{poly}(\frac{1}{\epsilon},m,N)$.\footnote{Notice that a best response can be found in constant time if $\mH_i$ is of constant size. In addition, in{\ifnum\Includeappendix=1 { Section \ref{sec:bestresponsregression}}\else{ the appendix}\fi} we leverage the algorithm proposed in \cite{porat2017best}, and show that it can compute a best response within the set of linear predictors efficiently when the input dimension (denoted by $n$ in the model above) is constant. We also discuss situations where a better response cannot be computed efficiently in Section \ref{sec:discussion}, and present the applicability of our models for these cases as well.}

\section{Meta algorithm and analysis}
\label{sec:meta}
Throughout this section we assume the pseudo-dimension of $\mH_i$ is finite,  and we denote it by $d_i$, i.e. $\pdim(\mH_i)=d_i<\infty$. Our goal is to propose a generic method for finding an $\epsilon$-PNE efficiently. 
The method is composed of two steps: first, it attains a sample of ``sufficient'' size. Afterwards, it runs an $\epsilon$-better-response dynamics until convergence, and returns the obtained profile. The underlying idea is straightforward, but its analysis is non-trivial. In particular, we need to show two main claims: 
\begin{itemize}
\item Given a sufficiently large sample $\mS$, the payoff of each player $i$ in the empirical game is not too far away from her payoff in the actual game, with high probability. This holds concurrently for all possible strategy profiles.
\item An $\epsilon$-PNE exists in every empirical game. Therefore, players can reach an $\epsilon$-PNE of the empirical game fast, using their $\epsilon$-better-response oracles.
\end{itemize}
These claims will be made explicit in forthcoming Subsections \ref{subsec:sample} and \ref{subsec:dyn}. We formalize the above discussion via Algorithm \ref{algorithm:betterres} in Subsection \ref{subsec:alg}.

\subsection{Uniform convergence in probability}
\label{subsec:sample}
We now bound the probability (over all choices of $\mS$) of having player $i$'s payoff (for an arbitrary $i$) greater or less than its empirical counterpart by more than $\epsilon$. Notice that the restriction of $\mH_i$ to any arbitrary sample $\mS$, i.e. $\mH_i \cap \mS$, may be of infinite size. 
Nevertheless, the payoff function concerns the indicator function $\mI$ only and not the real-valued prediction produced by functions in $\mH_i$; therefore, we now analyze this binary function class.

Let $\mF_i:\mZ \rightarrow \{0,1\}$ such that 
\begin{equation}
\label{eq:defoff}
\mF_i \defeq \left\{\mI(z , h) \mid h\in \mH_i  \right\}.
\end{equation}
Notice that $\abs{\mF_i \cap \mS}$ represents the effective size of $\mH_i \cap \mS$ with respect to the indicator function $\mI$. 

We already know that the pseudo-dimension of $\mH_i$ is $d_i$. In Lemma \ref{lemma:regtoclass} we bind the pseudo-dimension of $\mH_i$ with the VC dimension of $\mF_i$.

\begin{lemma}
\label{lemma:regtoclass}
$\vc(\mF_i) \leq 10d_i$.
\end{lemma}

After discovering the connection between the growth rate of $\mH_i$ and $\mF_i$, we can progress to bounding the growth of the payoff function class $\mF$ (which we will define shortly).

For ease of notation, denote $\mI(z,\bl h)=(\mI(z,h_1),\dots ,\mI(z,h_N))$. Similarly, let $w(z;\bl h)=(w_1(z;\bl h),\dots ,w_N(z;\bl h))$.
Note that there is a bijection $\mI(z,\bl h) \mapsto w(z;\bl h)$, which divides $\mI(z,\bl h)$ by its norm if it is greater than zero or leaves it as is otherwise. Formally, there is a bijection $M$, $M:\{0,1\}^N \rightarrow \{1,\frac{1}{2},\dots,\frac{1}{N},0\}^N$ such that for every $\bl v \in \{0,1\}^N$,
\[
M(\bl v) = 
\begin{cases}
\bl 0 & \text{if } \norm{\bl v}=0\\
\frac{\bl v }{\norm{\bl v}} &\text{otherwise}
\end{cases}.
\]

Let $\mF=\mZ \rightarrow \{0,1 \}^N$, defined by
\[
\mF \defeq \left\{ \mI\left(z,\bl h  \right) \mid \bl h\in \mH\right\}.
\]
Note that every element in $\mF$ is a function from $\mZ$ to $\{0,1\}^N$. The restriction of $\mF$ to a sample $\mS$ is defined by
\[
\mF \cap \mS = \left\{ \left(\mI(z_1,\bl h),\dots,\mI(z_m,\bl h)   \right) \mid \bl h \in \mH   \right\}.
\]
Due to the aforementioned bijection, every element in $\mF \cap \mS$ represents a distinct payoff vector of the empirical game; thus, bounding $\abs{\mF \cap \mS}$ corresponds to bounding the number of distinct strategy profiles in the empirical game. Clearly,
\[
\abs{\mF \cap \mS}=\prod_{i=1}^N \abs{\mF_i \cap \mS}.
\]
The growth function of $\mF$, $\Pi_\mF (m) = \max_{\mS\in \mZ^m    } \abs{ \mF \cap \mS   }$, is therefore bounded as follows.
\begin{lemma}
\label{lemma:growthsumvc}
$\Pi_\mF(m) \leq (em)^{10\sum_{i=1}^N d_i}$.
\end{lemma}

Next, we bound the probability of a player $i$'s payoff being ``too far'' from its empirical counterpart. The proof of Lemma \ref{lemma:uniconvergenceoneplayer} below goes along the path of \citeauthor{vapnik2015uniform}, introduced in  \cite{vapnik2015uniform}. Since in our case $\mF$ is not a binary function class,  few modifications are needed. 
\begin{lemma} 
\label{lemma:uniconvergenceoneplayer}
Let $m$ be a positive integer, and let $\epsilon>0$. It holds that
\[
\Pr_{\mS\sim \mD^m}\left(\exists \bl{h} : \abs{\pi_i(\bl h)-\pi_i^{\mS}(\bl h)} \geq  \epsilon \right) \leq 4 \Pi_{\mF}(2m)e^{-\frac{\epsilon^2 m}{8}}.
\]
\end{lemma}
The following Theorem \ref{thm:unionbound} bounds the probability that any player $i$ has a difference greater than $\epsilon$ between its payoff and its empirical payoff (over the selection of a sample $\mS$), uniformly over all possible strategy profiles. This is done by simply applying the union bound on the bound already obtained in Lemma \ref{lemma:uniconvergenceoneplayer}.
\begin{theorem}
\label{thm:unionbound}
Let $m$ be a positive integer, and let $\epsilon>0$. It holds that
\begin{equation}
\label{eq:inthmunionbound}
\Pr_{\mS\sim \mD^m}\left(\exists i\in[N] : \sup_{\bl h \in \mH}\abs{\pi_i(\bl h)-\pi_i^{\mS}(\bl h)} \geq  \epsilon \right) \leq 4N (2em)^{10\sum_{i=1}^N d_i}e^{-\frac{\epsilon^2 m}{8}}.
\end{equation}
\end{theorem}

\subsection{Existence of a PNE in empirical games}
\label{subsec:dyn}
In the previous subsection we bounded the probability of a payoff vector being too far from its counterpart in the empirical game. Notice, however, that this result implies nothing about the existence of a PNE or an approximate PNE: for a fixed $\mS$, even if $\sup_{\bl h \in \mH}\abs{\pi_i(\bl h)-\pi_i^{\mS}(\bl h)} < \epsilon$ holds for every $i$, a player may still have a beneficial deviation. Therefore, the results of the previous subsection are only meaningful if we show that there exists a PNE in the empirical game, which is the goal of this subsection.
We prove this existence using the notion of {\em potential games} \cite{monderer1996potential}.

A non-cooperative game is called a potential game if there exists a function $\Phi:\mH \rightarrow\mathbb R$ such that for every strategy profile $\bl h=(h_1,\dots,h_N) \in \mH$ and every $i\in [N]$, whenever player $i$ switches from $h_i$ to a strategy $h_i'\in \mH_i$, the change in her payoff function equals the change in the potential function, i.e.
\[
\Phi (h'_{{i}},\bl h_{{-i}})-\Phi (h_{{i}},\bl h_{{-i}})=\pi_{{i}}(h'_{{i}},\bl h_{{-i}})-\pi_{{i}}(h_{{i}},\bl h_{{-i}}).
\]
\begin{theorem}[\cite{monderer1996potential,rosenthal1973class}] 
\label{thm:pot}
Every potential game with a finite strategy space possesses at least one PNE.
\end{theorem}
Obviously, in our setting the strategy space of a game instance $\langle \mZ,\mD,\mN,(\mH_i)_{i\in \mN} \rangle$ is typically infinite.
Infinite potential games may also possess a PNE (as discussed in \cite{monderer1996potential}), but in our case the distribution $\mD$ is approximated from samples and the empirical game is finite, so no stronger claims are needed.

Lemma \ref{lemma:PNEexistence} below shows that every empirical game is a potential game.
\begin{lemma}
\label{lemma:PNEexistence}
Every empirical game  $\langle \mZ,\mS \sim \mD^m ,\mN,(\mH_i)_{i\in \mN} \rangle$ has a potential function.
\end{lemma}
As an immediate result of Theorem \ref{thm:pot} and Lemma \ref{lemma:PNEexistence}, 
\begin{corollary}
\label{corollary:PNEexistence}
Every empirical game  $\langle \mZ,\mS \sim \mD^m ,\mN,(\mH_i)_{i\in \mN} \rangle$ possesses at least one PNE.
\end{corollary}

After establishing the existence of a PNE in the empirical game, we are interested in the rate with which it can be ``learnt''. More formally, we are interested in the convergence rate of the dynamics between the players, where at every step one player deviates to one of her $\epsilon$-better responses. Such dynamics do not necessarily converge in general games, but do converge in potential games. By examining the specific potential function in our class of (empirical) games, we can also bound the number of steps until convergence.

\begin{lemma}
\label{lemma:betterconverge}
Let  $\langle \mZ,\mS \sim \mD^m ,\mN,(\mH_i)_{i\in \mN} \rangle$ be any empirical game instance. After at most $O\left(\frac{\log N}{\epsilon}  \right)$ iterations of any $\epsilon$-better-response dynamics, an $\epsilon$-PNE is obtained.
\end{lemma}
\subsection{Learning $\epsilon$-PNE with high probability}
\label{subsec:alg}
In this subsection we leverage the results of the previous Subsections \ref{subsec:sample} and \ref{subsec:dyn} to devise Algorithm \ref{algorithm:betterres}, which runs in polynomial time and returns an approximate equilibrium with high probability. More precisely, we show that Algorithm \ref{algorithm:betterres} returns an $\epsilon$-PNE with probability of at least $1-\delta$, and has time complexity of $\text{poly}\left(\frac{1}{\epsilon},m,N,\log\left( \frac{1}{\delta}\right),d \right)$. As in the previous subsections, we denote $d=\sum_{i=1}^N d_i $.

First, we bound the required sample size. Using standard algebraic manipulations on Equation (\ref{eq:inthmunionbound}), we obtain the following.
\begin{lemma}
\label{lemma:deltaandm}
Let $\epsilon,\delta \in (0,1)$, and let
\begin{equation}
\label{eq:lemma:deltaandm}
m \geq \frac{320d}{\epsilon^2} \log\left( \frac{160d}{\epsilon^2} \right) +\frac{160d\log(2e)}{\epsilon^2}+\frac{16}{\epsilon^2}\log\left( \frac{4N}{\delta} \right).
\end{equation}
With probability of at least $1-\delta$ over all possible samples $\mS$ of size $m$, it holds that
\[
\forall i\in [N] : \sup_{\bl h \in \mH}\abs{\pi_i(\bl h)-\pi_i^{\mS}(\bl h)} <  \epsilon.
\]
\end{lemma}
Given $\epsilon,\delta$, we denote by $m_{\epsilon,\delta}$ the minimal integer $m$ satisfying Equation (\ref{eq:lemma:deltaandm}). Lemma \ref{lemma:deltaandm} shows that $m_{\epsilon,\delta}=O\left(\frac{d}{\epsilon^2}\log\left( \frac{d}{\epsilon^2} \right) + \frac{1}{\epsilon^2}\log\left( \frac{N}{\delta} \right) \right)$ are enough samples to have all empirical payoff vectors $\epsilon$-close to their theoretic counterpart coordinate-wise (i.e. $L^\infty$ norm), with a probability of at least $1-\delta$. 

Next, we bind an approximate PNE in the empirical game with an approximate PNE in the (actual) game.
\begin{lemma}
\label{lemma:empiseq}
Let  $m \geq m_{\frac{\epsilon}{4},\delta}$ and let $\bl h$ be an $\frac{\epsilon}{2}$-PNE in  $\langle \mZ,\mS \sim \mD^m ,\mN,(\mH_i)_{i\in \mN} \rangle$. Then $\bl h$ is an $\epsilon$-PNE with probability of at least $1-\delta$.
\end{lemma}

Recall that Lemma \ref{lemma:betterconverge} ensures that every $O\left(\frac{\log N}{\epsilon} \right)$ iterations of any $\epsilon$-better-response dynamics must converge to an $\epsilon$-PNE of the empirical game. In each such iteration a player calls her approximate better-response oracle, which is assumed to run in $\text{poly}(\frac{1}{\epsilon},m,N)$ time. Altogether, given $\epsilon$ and $\delta$, Algorithm \ref{algorithm:betterres} runs in $\text{poly}\left(\frac{1}{\epsilon},N,\log\left( \frac{1}{\delta}\right),d \right)$ time, and returns an $\epsilon$-PNE with probability of at least $1-\delta$.

\begin{algorithm}[t]
\DontPrintSemicolon
\caption{ Approximate PNE w.h.p. via better-response dynamics\label{algorithm:betterres}}
\KwIn{$\delta, \epsilon \in (0,1)$
}
\KwOut{a strategy profile $\bl h$}
set $m = m_{\frac{\epsilon}{2},\delta}$  \tcp*{the minimal integer $m$ satisfying Equation (\ref{eq:lemma:deltaandm})}
sample $\mS$ from $\mD^m$\;
execute any $\frac{\epsilon}{2}$-better-response dynamics on $\mS$ until convergence, and obtain a strategy profile $\bl h$ that is an empirical $\frac{\epsilon}{2}$-PNE \;
\Return{$\bl h$}
\end{algorithm}

\section{Learnability in games with infinite dimension}
\label{sec:exampleandinf}


While Lemma \ref{lemma:regtoclass} upper bounds $\vc(\mF_i)$ as a function of $\pdim(\mH_i)$, it is fairly easy to show that $\vc(\mF_i)\geq \pdim(\mH_i) $ ({\ifnum\Includeappendix=1{see Claim \ref{claim:flowerbound} in the appendix}\else{we prove this claim formally in the appendix}\fi}). Therefore, if $\pdim(\mH_i)$ is infinite, so is $\vc(\mF_i)$.

Classical results in learning theory suggest that if $\pdim(\mF_i)=\infty$, a best response on the sample may not generalize to an approximate best response w.h.p. To see this, imagine a ``game'' with one player, who seeks to maximize her payoff function. No Free Lunch Theorems (see, e.g., \cite{wolpert1997no}) imply that with a constant probability the player cannot get her payoff within a constant distance from the optimal payoff. We conclude that in general games, if a player has a strategy space with an infinite pseudo-dimension, she may not be able to learn. However, in the presence of such a player, can other players with a finite pseudo-dimension learn an approximate best-response?


One typically shows non-learnability by constructing two distributions and proving that with constant probability an agent cannot tell which distribution produced the sample she obtained. These two distributions are constructed to be distant enough from each other, so the loss (or payoff in our setting) is far from optimal by  at least a constant. In our setting, however, players are interacting with each other, and player payoffs are a function of the whole strategy profile; thus, interesting phenomena occur even if the distribution $\mD$ is known. In particular,  Example \ref{example:unlearnability} below demonstrates that in the infinite dimension case, not every empirical PNE is generalized to an approximate PNE with high probability. 

\begin{example}
\label{example:unlearnability}
Let $\mD$ be a density function over $\mZ=[0,2] \times \{0,1\} \times \left\{  \frac{1}{2}\right\}$ as follows:
\[
\mD(x,y,t) = 
\begin{cases}
\frac{1}{2} & 0\leq x<1,y=0, t=\frac{1}{2}\\
\frac{1}{2} & 1\leq x\leq 2,y=1, t=\frac{1}{2}\\
0 & \text{otherwise}
\end{cases}.
\]
In addition, for any finite size subset $\mS$  of $\mZ$ in the support of $\mD$, denote
\[
h^{\mS\rightarrow 0}(x) = 
\begin{cases}
0 & \exists y,t: (x,y,t)\in \mS \\
\ind_{1\leq x\leq 2} & \forall y,t: (x,y,t)\notin   \mS 
\end{cases}, \quad
h^{\mS\rightarrow 1}(x) = 
\begin{cases}
1 & \exists y,t: (x,y,t)\in \mS \\
\ind_{1\leq x\leq 2} & \forall y,t: (x,y,t)\notin  \mS 
\end{cases}.
\]
In other words, $h^{\mS\rightarrow 0}$ labels 0 every instance that appears in the sample $\mS$ and every instance in the $[0,1)$ segment. On the other hand, $h^{\mS\rightarrow 1}$ labels 1 every instance that appears in the sample $\mS$ and every instance in the $[1,2]$ segment. Denote
\[
\mH_1 = \{h^{\mS\rightarrow 0} \mid \mS \subset \mZ   \}   \cup \{h^{\mS\rightarrow 1} \mid \mS \subset \mZ   \},
\]
and let $\mH_2 = \mH_3 = \{ \ind_{0\leq x < 1}, \ind_{1\leq x\leq 2}  \}$. In this three-player game, consider the profile $\bl h=(h_1,h_2,h_3)$ such that
$ h_1= h^{\mS\rightarrow 0},\quad h_2=h_3 =\ind_{1\leq x\leq 2}$. 
Notice that the payoffs under $\bl h$ are defined as follows:
\[
\pi_1(\bl h )=\frac{1}{m}\sum_{j=1}^m (1-y_j),\quad \pi_2(\bl h )=\pi_2(\bl h )=\frac{1}{2m}\sum_{j=1}^m y_j.
\]
Observe that if $\frac{1}{2}<\frac{1}{m}\sum_{j=1}^m y_j < \frac{3}{4}$, then $\bl h$ is an empirical PNE, since no player can improve her payoff. Notice, however, that $\pi_3(\bl h)=\frac{1}{6}$ yet $\pi_3(\ind_{0\leq x < 1},\bl h_{-2})=\frac{1}{4}$.

Since we have $\frac{1}{2}<\frac{1}{m}\sum_{j=1}^m y_j < \frac{3}{4}$ with probability of at least $\frac{1}{4}$ over all choices of $\mS$ for $\abs{\mS} \geq 15$ (see{\ifnum\Includeappendix=1{ Claim \ref{claim:auxhoeffbin} in}\fi} the appendix), this empirical equilibrium will not be generalized to $\frac{1}{12}$-PNE w.p. of at least $\frac{1}{4}$. This is true for any $\epsilon,\delta \in(0,1)$; thus, an empirical PNE is not generalized to an approximate PNE w.h.p. 
\end{example}

Another interesting point is that in Example \ref{example:unlearnability} each player can trivially find a strategy that maximizes her payoff if she were alone, since $\mD$ is known. Indeed, this inability to generalize from samples follows solely from strategic behavior. Notice that if player 3 has knowledge of $\mH_1$, she can infer that her strategy under $\bl h$ is sub-optimal. However, knowledge of the strategy spaces of other players is a heavy assumption: the better-response dynamics we discussed in Subsection \ref{subsec:dyn} only assumed that each player can compute a better response.


\section{Discussion}
\label{sec:discussion}

As mentioned in Section \ref{subsec:model}, our analysis assumes players have better-response oracles. In fact, our model and results are valid for a much more general scenario, as described next. Consider the case where players only have heuristics for finding a better response. After running heuristic better-response dynamics and obtaining a strategy profile, the payoffs with respect to the whole population are guaranteed to be close to their empirical counterparts, w.h.p.; therefore, our analysis is still meaningful even if players cannot maximize their empirical payoff efficiently, as the bounds on the required sample size we obtained in Section \ref{sec:meta} and the rate of convergence are relevant for this case as well. 

The reader may wonder about a variation of our model, where player payoffs are defined differently. For example, consider each user as granting one monetary unit to the player that offers the closest prediction to his instance. This definition is in the spirit of Dueling Algorithms \cite{immorlica2011dueling} and Best Response Regression \cite{porat2017best}. Under this payoff function, and unlike our model, an empirical PNE does not necessarily exist. Nevertheless, we believe that examining and understanding these scenarios is fundamental to analysis of competing prediction algorithms, and deserves future work.

{\ifnum\Includeackerc=1{
\section*{Acknowledgments}\label{sec:Acknowledgments}
This project has received funding from the European Research Council (ERC) under the European Union's Horizon 2020 research and innovation programme (grant agreement n$\degree$  740435).}\fi}

{\ifnum\Includeappendix=1{ 

\appendix

\section{Omitted proofs from Section \ref{sec:meta}}

\begin{proof}[\textbf{Proof of Lemma \ref{lemma:regtoclass}}]
First, we define two auxiliary classes of binary functions $\mG^{\geq},\mG^{\leq}$ such that
\begin{align}
\label{eq:ggeqeqdef}
&\mG^{\geq} =\{g_h^\geq (x,r)= \ind_{ h(x) \geq r} \mid h\in \mH_i ,(x,r)\in \mX\times \mathbb{R}  \},\nonumber \\
&\mG^{\leq} =\{g_h^\leq (x,r)=\ind_{ h(x) \leq r} \mid h\in \mH_i ,(x,r)\in \mX\times \mathbb{R}  \}.
\end{align}
\begin{claim}
\label{claim:auxiliaryg}
$\vc(\mG^{\geq})=\vc(\mG^{\leq}) = d_i$. 
\end{claim}
The proof of Claim \ref{claim:auxiliaryg} appears in Section \ref{sec:additionalproofs}. Next, we wish to bound the growth function of $\mF_i$ using the growth function of $\mG^{\geq}$ and $\mG^{\leq}$.
\begin{claim}
\label{claim:fgrowth}
$\Pi_{\mF_i}(m) \leq \Pi_{\mG^{\geq}}(m) \cdot \Pi_{\mG^{\leq}}(m)$.
\end{claim}
The proof of Claim \ref{claim:fgrowth} appears in Section \ref{sec:additionalproofs}. We are now ready for the final argument. By the Sauer-Shelah lemma we know that every $m$ satisfying $2^m > \Pi_{\mF_i}(m)$ is an upper bound on $\vc(\mF_i)$ \cite{sauer1972density}. In particular, for $m=10d_i$ we have
\begin{align}
\sqrt{2^{10d_i}}&=2^{5d_i}=(31+1)^{d_i}=\sum_{j=0}^{d_i}31^j 1^{d_i-j} {d_i \choose j} \stackrel {\text{Claim \ref{claim:binomialfactors}}}{\geq }
\sum_{j=0}^{d_i} \left( \frac{31 d_i}{j}\right)^j  \\
&\geq \sum_{j=0}^{d_i} \left( \frac{10 e d_i}{j}\right)^j \stackrel {\text{Claim \ref{claim:binomialfactors}}}{\geq } \sum_{j=0}^{d_i} {10d_i \choose j} \nonumber \\ 
&\geq \Pi_{\mG^{\geq}}(10d_i)=\Pi_{\mG^{\leq}}(10d_i);
\end{align}
therefore
\[
\Pi_{\mF_i}(10d_i) \leq \Pi_{\mG^{\geq}}(10d_i)\Pi_{\mG^{\leq}}(10d_i) < 2^{10d_i} .
\]

\end{proof}

\begin{proof}[\textbf{Proof of Lemma \ref{lemma:growthsumvc}}]
Recall that the Sauer - Shelah lemma implies that $\Pi_{\mF_i}(m) \leq (em/d_i)^{d_i}$ for $m>d_i+1$. Since $\abs{\mF \cap \mS}=\prod_{i=1}^N \abs{\mF_i \cap \mS}$, we have
\[
\Pi_{\mF}(m) \leq \prod_{i=1}^N \Pi_{\mF_i}(m) \leq \prod_{i=1}^N (em/d_i)^{10d_i}\leq  \prod_{i=1}^N (em)^{10d_i}=(em)^{10\sum_{i=1}^N d_i}
\]
\end{proof}

\begin{proof}[\textbf{Proof of Lemma \ref{lemma:uniconvergenceoneplayer}}]
\newcommand{\hbad}{\tilde{ \bl h}(\mS)}
The proof follow closely the four steps in the proof of the classical uniform convergence theorem for binary functions (see, e.g., \cite{anthony2009neural,vapnik1971uniform}). The only steps that need modification are step 3 and 4, but we present the full proof for completeness.

\textit{Step 1 -- Symmetrization:}
First, we want to show that
\begin{equation}
\Pr_{\mS\sim \mD^m}\left(\exists \bl{h} : \abs{\pi_i(\bl h)-\pi_i^{\mS}(\bl h)} \geq  \epsilon \right) \leq 
2 \Pr_{(\mS,\mS')\sim \mD^m}\left(\exists \bl{h} : \abs{\pi_i^{\mS}(\bl h)-\pi_i^{\mS'}(\bl h)} \geq  \frac{\epsilon}{2} \right) \label{eq:steponeaa}.
\end{equation}

 For each $\mS$, let $\hbad$ be a function for which $\abs{\pi_i(\hbad)-\pi_i^{\mS}(\hbad)} \geq \epsilon$ if such a function exists, and any other fixed function in $\mH$ otherwise. Notice that if $\abs{\pi_i(\hbad)-\pi_i^{\mS}(\hbad)} \geq \epsilon$ and  $\abs{\pi_i(\hbad)-\pi_i^{\mS'}(\hbad)} \leq  \frac{\epsilon}{2}$, then $\abs{\pi_i^{\mS}(\hbad)-\pi_i^{\mS'}(\hbad)} \geq   \frac{\epsilon}{2}$ (triangle inequality); thus,
\begin{align*}
&\Pr_{(\mS,\mS')\sim \mD^m}\left(\exists \bl{h} : \abs{\pi_i^{\mS}(\bl h)-\pi_i^{\mS'}(\bl h)} \geq  \frac{\epsilon}{2} \right) \\
&\geq \Pr_{(\mS,\mS')\sim \mD^m}\left(  \abs{\pi_i^{\mS}(\hbad)-\pi_i^{\mS'}(\hbad)} \geq  \frac{\epsilon}{2} \right) \\
&\geq \Pr_{(\mS,\mS')\sim \mD^m}\left( \abs{\pi_i(\hbad)-\pi_i^{\mS}(\hbad)} \geq \epsilon \cap \abs{\pi_i(\hbad)-\pi_i^{\mS'}(\hbad)}  \leq \frac{\epsilon}{2} \right) \\
&=\E_{\mS\sim \mD^m} \left[ \ind\left( \abs{\pi_i(\hbad)-\pi_i^{\mS}(\hbad)} \geq \epsilon  \right)\Pr_{\mS'\mid \mS}\left( \abs{\pi_i(\hbad)-\pi_i^{\mS'}(\hbad)}  \leq \frac{\epsilon}{2}    \right) \right] \\ 
&\geq \frac{1}{2}\Pr_{\mS \sim \mD^m}\left(  \abs{\pi_i(\hbad)-\pi_i^{\mS}(\hbad)} \geq \epsilon \right)\\
&= \frac{1}{2}\Pr_{\mS \sim \mD^m}\left(  \exists \bl{h} : \abs{\pi_i(\bl h)-\pi_i^{\mS}(\bl h)} \geq \epsilon \right),
\end{align*}

since $\mS,\mS'$ are independent and due to Claim \ref{claim:rqep}.

\textit{Step 2 -- Permutations:} we denote $\Gamma_{2m}$ as the set of all permutations of $[2m]$ that swap $i$ and $m+i$ in some subset of $[m]$. Namely,
\[
\Gamma_{2m} = \{\sigma\in  \Pi([2m]) \mid \forall i\in[m] :\sigma(i) = i \lor \sigma(i)= m+i; \forall i,j \in [2m]: \sigma(i) = j \Leftrightarrow \sigma(j)=i   \},
\]
where $\Pi([2m])$ denotes the set of permutations over $[2m]$. In addition, for $\mS = (z_1,\dots,z_{2m})$, let $\sigma(\mS)=(z_{\sigma(1)},\dots,z_{\sigma(2m)})$. Notice that for every $\sigma \in \Gamma_{2m}$ it holds that
\[
\Pr_{(\mS,\mS')\sim \mD^m}\left(\exists \bl{h} : \abs{\pi_i^{\mS}(\bl h)-\pi_i^{\mS'}(\bl h)} \geq  \frac{\epsilon}{2} \right) = \Pr_{(\mS,\mS')\sim \mD^m}\left(\exists \bl{h} : \abs{\pi_i^{\sigma(\mS)}(\bl h)-\pi_i^{\sigma(\mS')}(\bl h)} \geq  \frac{\epsilon}{2} \right);
\]
hence,
\begin{align}
&\Pr_{(\mS,\mS')\sim \mD^{2m}}\left(\exists \bl{h} : \abs{\pi_i^{\mS}(\bl h)-\pi_i^{\mS'}(\bl h)} \geq  \frac{\epsilon}{2} \right) \nonumber\\
& =\frac{1}{2^m}\sum_{\sigma\in \Gamma_{2m}} \Pr_{(\mS,\mS')\sim \mD^{2m}}\left(\exists \bl{h} : \abs{\pi_i^{\sigma(\mS)}(\bl h)-\pi_i^{\sigma(\mS')}(\bl h)} \geq  \frac{\epsilon}{2} \right)\nonumber\\
& =\frac{1}{2^m}\sum_{\sigma\in \Gamma_{2m}} \E_{(\mS,\mS')\sim \mD^{2m}}\left[\ind_{\exists \bl{h} : \abs{\pi_i^{\sigma(\mS)}(\bl h)-\pi_i^{\sigma(\mS')}(\bl h)} \geq  \frac{\epsilon}{2}} \right]\nonumber\\
& =\E_{(\mS,\mS')\sim \mD^{2m}}\left[\frac{1}{2^m}\sum_{\sigma\in \Gamma_{2m}} \ind_{\exists \bl{h} : \abs{\pi_i^{\sigma(\mS)}(\bl h)-\pi_i^{\sigma(\mS')}(\bl h)} \geq  \frac{\epsilon}{2}} \right]\nonumber\\
& =\E_{(\mS,\mS')\sim \mD^{2m}}\left[\Pr_{\sigma\in \Gamma_{2m}} \left(\exists \bl{h} : \abs{\pi_i^{\sigma(\mS)}(\bl h)-\pi_i^{\sigma(\mS')}(\bl h)} \geq  \frac{\epsilon}{2}\right) \right]\nonumber\\
&\leq \sup_{(\mS,\mS')\sim \mD^{2m}}\left[\Pr_{\sigma\in \Gamma_{2m}} \left(\exists \bl{h} : \abs{\pi_i^{\sigma(\mS)}(\bl h)-\pi_i^{\sigma(\mS')}(\bl h)} \geq  \frac{\epsilon}{2}\right) \right] \label{eq:steptwoperm}.
\end{align}

\textit{Step 3 -- Reduction to finite class:} fix $(\mS,\mS')$ and consider a random draw of $\sigma \in \Gamma_{2m}$. For each strategy profile $\bl h$, the quantity $\abs{\pi_i^{\sigma(\mS)}(\bl h)-\pi_i^{\sigma(\mS')}}$ is a random variable. Since $\mF\cap \mS$ represents the number of distinct strategy profiles in the empirical game over $\mS$ (see Subsection \ref{subsec:sample}), there are at most $\Pi_{\mF}(2m)$ such random variables.

\begin{align}
&\Pr_{\sigma\in \Gamma_{2m}} \left(\exists \bl{h} : \abs{\pi_i^{\sigma(\mS)}(\bl h)-\pi_i^{\sigma(\mS')}(\bl h)} \geq  \frac{\epsilon}{2}\right)  \nonumber\\
&\leq \Pi_{\mH}(2m) \sup_{\bl h \in \mH} \Pr_{\sigma\in \Gamma_{2m}} \left(\abs{\pi_i^{\sigma(\mS)}(\bl h)-\pi_i^{\sigma(\mS')}(\bl h)} \geq  \frac{\epsilon}{2}\right) \label{hehehehe}
\end{align}

\textit{Step 4 -- Hoeffding's inequality:} by viewing the previous equation as Rademacher random variables, we have
\begin{align*}
&\Pr_{\sigma\in \Gamma_{2m}} \left(\abs{\pi_i^{\sigma(\mS)}(\bl h)-\pi_i^{\sigma(\mS')}(\bl h)} \geq  \frac{\epsilon}{2}\right) \\
&= \Pr_{\sigma\in \Gamma_{2m}} \left(\frac{1}{m}\abs{\sum_{j=1}^m w_i(z_{\sigma(j)},\bl h)-w_i(z_{\sigma(j+m)},\bl h)} \geq  \frac{\epsilon}{2}\right)\\
& =\Pr_{r\in\{-1,1\}^m}\left(\frac{1}{m}\abs{\sum_{j=1}^m r_j\left(w_i(z_j,\bl h)-w_i(z_{j+m},\bl h)\right)} \geq  \frac{\epsilon}{2}\right).
\end{align*}
Observe that for every $j$ it holds that $r_j\left(w_i(z_j,\bl h)-w_i(z_{j+m},\bl h)\right) \in [-1,1]$, and
\[
\E_{r_j\in\{-1,1\}}\left[ r_j\left(w_i(z_j,\bl h)-w_i(z_{j+m},\bl h)\right) \right]=0
\]
holds due to symmetry. By applying Hoeffding's inequality we obtain
\begin{align}
\label{eq:sfourfin}
 \Pr_{r\in\{-1,1\}^m}\left(\frac{1}{m}\abs{\sum_{j=1}^m r_j\left(w_i(z_j,\bl h)-w_i(z_{j+m},\bl h)\right)} \geq  \frac{\epsilon}{2}\right) \leq 2e^{-\frac{m\epsilon^2}{8}}.
\end{align}
Finally, we combining Equations (\ref{eq:steponeaa}),(\ref{eq:steptwoperm}),(\ref{hehehehe}) and (\ref{eq:sfourfin}) we derive the desired result.
\end{proof}

\begin{proof}[\textbf{Proof of Theorem \ref{thm:unionbound}}]
The Theorem follows immediately by applying the union bound on the inequality obtained in Lemma \ref{lemma:uniconvergenceoneplayer} and by substituting $\Pi_\mF(2m) $ according to Lemma \ref{lemma:growthsumvc}.
\end{proof}

\begin{proof}[\textbf{Proof of Lemma \ref{lemma:PNEexistence}}]
The lemma is proven by showing that induced game has a potential function $\Phi : \mH \rightarrow \mathbb R$. Namely, we show a function $\Phi$ such that for every $i,\bl h$ and $h_i'$ it holds that \[ \pi_i(\bl h)-\pi_i(h_i', \bl h_{-i}) = \Phi(\bl h)-\Phi(h_i', \bl h_{-i}). \]
Denote $\Phi(\bl h) = \frac{1}{m}\sum_{j=1}^m \sum_{k=1}^{N(z_j;\bl h)} \frac{1}{k}$, and observe that
\begin{align*}
& \pi_i(\bl h)-\pi_i(h_i', \bl h_{-i})=\frac{1}{m}\sum_{j=1}^m w_i(z_j;\bl h) -   \frac{1}{m}\sum_{j=1}^m w_i(z_j;h_i', \bl h_{-i}) \\
& =\frac{1}{m}\sum_{j=1}^m \frac{\mI_i(z;h_i)}{\abs{N(z;\bl h)}}  -   \frac{1}{m}\sum_{j=1}^m  \frac{\mI_i(z;h_i')}{\abs{N(z;h_i', \bl h_{-i})} }+\frac{1}{m}\sum_{j=1}^m \sum_{k=1}^{N(z_j;\bl h_{-i})} \frac{1}{k}-\frac{1}{m}\sum_{j=1}^m \sum_{k=1}^{N(z_j;\bl h_{-i})} \frac{1}{k} \\
& = \frac{1}{m}\sum_{j=1}^m \sum_{k=1}^{N(z_j;\bl h)} \frac{1}{k} - \frac{1}{m}\sum_{j=1}^m \sum_{k=1}^{N(z_j;h_i, \bl h_{-i})} \frac{1}{k}=\Phi(\bl h)-\Phi(h_i', \bl h_{-i}).
\end{align*}
\end{proof}

\begin{proof}[\textbf{Proof of Lemma \ref{lemma:betterconverge}}]
In each iteration of the dynamics it holds that
\begin{align}
\label{eq:bnfesw}
\Phi(h_i', \bl h_{-i}) - \Phi(\bl h) = \pi_i(h_i', \bl h_{-i}) - \pi_i(\bl h) \geq \epsilon.
\end{align}
Notice that 
\begin{equation}
\label{eq:aewrbg}
\Phi(\bl h) =  \frac{1}{m}\sum_{j=1}^m \sum_{k=1}^{N(z_j;\bl h)} \frac{1}{k} \leq  \frac{1}{m}\sum_{j=1}^m \sum_{k=1}^{N} \frac{1}{k} \leq \ln N +1.
\end{equation}
Since the potential is bounded by $\ln N +1$ and increases by at least $\epsilon$ per iteration throughout the dynamics, after at most $\frac{\ln N +1}{\epsilon}$ iterations it will reach its maximum value, thereby obtaining an $\epsilon$-PNE.
\end{proof}

\begin{proof}[\textbf{Proof of Lemma \ref{lemma:deltaandm}}]
By Equation (\ref{eq:inthmunionbound}), we look for $m$ that satisfies
\[
4N (2em)^{10\sum_{i=1}^N d_i}e^{-\frac{\epsilon^2 m}{8}} \leq \delta
\]
for given $\epsilon,\delta$; thus
\begin{align}
&(2em)^{10\sum_{i=1}^N d_i}e^{-\frac{\epsilon^2 m}{8}} \leq \frac{\delta}{4N} \nonumber \\
&\Rightarrow 10d \log(2em)-\frac{\epsilon^2 m}{8} \leq \log \frac{\delta}{4N}\nonumber \\
& \frac{\epsilon^2 m}{8} \geq 10d\log(2em)-\log \frac{\delta}{4N}\nonumber \\
& m\geq \frac{80d\log(2em)}{\epsilon^2}-\frac{8}{\epsilon^2}\log \frac{\delta}{4N}\nonumber \\
& m\geq \frac{80d}{\epsilon^2}\log(m)+\frac{80d\log(2e)}{\epsilon^2}+\frac{8}{\epsilon^2}\log \frac{4N}{\delta} \label{eq:dgfjnkindfjsnk}
\end{align}
Next, 
\begin{claim}[\cite{shalev2014understanding}, Section A]
\label{claim:fromsssappendix}
Let $a\geq 1$ and $b>0$. If $m \geq 4a\log(2a)+2b$, then $m\geq a \log m +b$.
\end{claim}
Set $a = \frac{80d}{\epsilon^2}$ and $b=\frac{80d\log(2e)}{\epsilon^2}+\frac{8}{\epsilon^2}\log \frac{4N}{\delta}$. Due to Claim \ref{claim:fromsssappendix}, we know that every $m$ that satisfies
\[
m \geq \frac{320d}{\epsilon^2} \log\left( \frac{160d}{\epsilon^2} \right) +\frac{160d\log(2e)}{\epsilon^2}+\frac{16}{\epsilon^2}\log\left( \frac{4N}{\delta} \right)
\]
also satisfy Equation (\ref{eq:dgfjnkindfjsnk}).
\end{proof}

\begin{proof}[\textbf{Proof of Lemma \ref{lemma:empiseq}}]
Notice that for every $i,h'_i$ it holds that
\begin{align*}
\pi_i(h_i', \bl h_{-i})-\pi_i(\bl h)&=\pi_i(h_i', \bl h_{-i})-\pi_i^\mS(h_i', \bl h_{-i})+\pi_i^\mS(h_i', \bl h_{-i})-\pi_i(\bl h) \\
& \stackrel{\substack{\bl h \text{ is an } \\ \frac{\epsilon}{2}\text{-empirical-PNE}} }{\leq} \pi_i(h_i', \bl h_{-i})-\pi_i^\mS(h_i', \bl h_{-i})+\pi_i^\mS(\bl h)-\pi_i(\bl h)+\frac{\epsilon}{2};
\end{align*}
therefore, if $\pi_i(h_i', \bl h_{-i})-\pi_i(\bl h)>\epsilon$ then at least one of $\pi_i(h_i', \bl h_{-i})-\pi_i^\mS(h_i', \bl h_{-i})> \frac{\epsilon}{4}$ or $\pi_i^\mS(\bl h)-\pi_i(\bl h) > \frac{\epsilon}{4}$ must hold. Overall,
\begin{align*}
&\Pr_{\mS\sim \mD^m}\left(\bl h \text{ is not an  } \epsilon \text{-PNE}  \right) = \Pr_{\mS\sim \mD^m}\left(\exists i \in [N], h'_i \in \mH_i: \pi_i(h_i', \bl h_{-i})-\pi_i(\bl h) > \epsilon \right) \\
& \leq \Pr_{\mS\sim \mD^m}\left(\exists i \in [N], h'_i \in \mH_i: \pi_i(h_i', \bl h_{-i})-\pi_i^\mS(h_i', \bl h_{-i})> \frac{\epsilon}{4} \text{ or } \pi_i^\mS(\bl h)-\pi_i(\bl h) > \frac{\epsilon}{4} \right) \\
& \leq \Pr_{\mS\sim \mD^m}\left(\exists i \in [N], h'_i \in \mH_i: \abs{\pi_i(h_i', \bl h_{-i})-\pi_i^\mS(h_i', \bl h_{-i})> \frac{\epsilon}{4}} \text{ or } \abs{\pi_i^\mS(\bl h)-\pi_i(\bl h) > \frac{\epsilon}{4}} \right) \\
& \leq \Pr_{\mS\sim \mD^m}\left(\exists i\in[N] : \sup_{\bl h'' \in \mH}\abs{\pi_i(\bl h'')-\pi_i^{\mS}(\bl h'')} \geq  \frac{\epsilon}{4} \right) \stackrel{m \geq m_{\frac{\epsilon}{4},\delta}}{\leq }\delta.
\end{align*}

\end{proof}

\section{Additional claims and proofs}
\label{sec:additionalproofs}
\begin{proof}[\textbf{Proof of Claim \ref{claim:auxiliaryg}}]

We prove the claim for $\mG^{\geq}$, and by symmetric arguments one can show it holds for $\mG^{\leq }$ as well.

Since $\pdim(\mH_i)=d_i$, for every $m\leq d_i$ there is a sample $\mS=(x_1,\dots  x_m)\in \mX^m$ and a witness $\bl r=(r_1,\dots r_m) \in \R^m$ such that for every binary vector $\bl b\in \{-1,1\}^m$ there is a function $h_{\bl b}$ for which $\sign(h_{\bl b}(x_j)-r_j)=b_j$ for all $j\in [m]$. Denote 
\[
\mS' = \left( (x_1,r_1),\dots,(x_m,r_m) \right),
\]
and focus on a particular $\bl b\in \{-1,1\}^m$. For every $j\in [m]$ such that $b_j=1$ we have
\[
\sign(h_{\bl b}(x_j)-r_j)=1 \Rightarrow  h_{\bl b}(x_j)-r_j >0 \Rightarrow\ind_{h_{\bl b}(x)\geq r}(x_j,r_j)=1.
\]
In addition, if $b_j=-1$ then
\[
\sign(h_{\bl b}(x_j)-r_j)=-1 \Rightarrow  h_{\bl b}(x_j)-r_j <0 \Rightarrow\ind_{h_{\bl b}(x)\geq r}(x_j,r_j)=0.
\]
This is true for every $\bl b$; therefore, we showed that $\mG^{\geq}$ shatters $\mS'$.

In the opposite direction, assume by contradiction that $\mG^{\geq}$ shatters $\mS'=\left( (x_1,r_1),\dots,(x_m,r_m) \right)$ for $m\geq d_i+1$. Let $H=\{ h_{\bl b} \}_{\bl b\in \{-1,1\}} \subset \mH_i$ be a set of functions such that for every $\bl b$ there exists exactly one function $h_{\bl b}\in H$ satisfying $\ind_{h_{\bl b}(x)\geq r}(x_j,r_j)=1$ if $b_j=1$ and $\ind_{h_{\bl b}(x)\geq r}(x_j,r_j)=0$ if $b_j=0$. Notice that by definition of VC dimension such $H$ must exist, and that $\abs{H}=2^m$.

One cannot claim directly that $\mH_i$ pseudo-shatters $\mS=(x_1,\dots,x_m)$ with witness $\bl r=(r_1,\dots r_m)$, since $h_{\bl b}(x_j)=r_j$ may hold for $b_j=1$, but we need $h_{\bl b}(x_j)$ to be strictly greater than $r_j$; therefore, we construct a new witness: let $a_j$ such that
\[
a_j= \max_{h_{\bl b}\in H, b_j=-1} h_{\bl b}(x_j).
\]
Notice that $\abs{H}$ is finite so the maximum is well defined. In addition, $a_j < r_j$ since $h_{\bl b}(x_j) <r_j$ for every $\bl b$ such that $b_j=-1$ (recall that if $b_j=0$ then  $\ind_{h_{\bl b}(x)\geq r}(x_j,r_j)=0$ ).

Denote $\bl r^* = \frac{\bl a +\bl r}{2}$. Next, we claim that $\mH_i$ pseudo-shatters $\mS$ with the witness $\bl r^*$. Fix $\bl b\in\{-1,1\}^m$. If $b_j=-1$,
\[
\ind_{h_{\bl b}(x)\geq r}(x_j,r_j)=0 \Rightarrow h_{\bl b}(x_j) \leq  a_j \Rightarrow h_{\bl b}(x_j) < r_j^* \Rightarrow \ind_{h_{\bl b}(x)\geq r}(x_j,r^*_j)=0.
\]

On the other hand, if $b_j=1$, we have
\[
\ind_{h_{\bl b}(x)\geq r}(x_j,r_j)=1 \Rightarrow h_{\bl b}(x_j) \geq  r_j \Rightarrow h_{\bl b}(x_j) > r_j^* \Rightarrow \ind_{h_{\bl b}(x)\geq r}(x_j,r^*_j)=1.
\]
Combining these two equations, we get that $\sign(h_{\bl b}(x_j)-r_j^*)=b_j$ for all $j\in [m]$. Consequently, $\mH_i$ pseudo-shatters $\mS$ with witness $\bl r^*$; hence we obtained a contradiction.

Overall, we showed that $\vc(\mG^{\geq}) \geq d_i$ and $\vc(\mG^{\geq}) \leq d_i$; hence $\vc(\mG^{\geq})=d_i$.
\end{proof}

\begin{proof}[\textbf{Proof of Claim \ref{claim:fgrowth}}]
Denote by $\mS=(x_j,y_j,t_j)_{j=1}^m \in \mZ^m$ an arbitrary sample, and let $\mF_i \cap \mS$ be the restriction of $\mF_i$ to $\mS$. Formally,
\[
\mF_i \cap \mS = \left\{(f_h(z_1),\dots,f_h(z_m))\mid f_h \in \mF_i  \right\}.
\] 

In addition, denote by $G^\geq $ the restriction of $\mG^\geq$ to $(x_j,y_j-t_j)_{j=1}^m$, and similarly let $G^\leq $ be the restriction of $\mG^\leq$ to $(x_j,y_j+t_j)_{j=1}^m$. We now show a one-to-one mapping $M: {\mF_i} \cap {\mS} \rightarrow G^\geq \times G^\leq$, implying that 
\begin{equation}
\label{eq:mnkasd}
\abs{{\mF_i} \cap {\mS} } \leq \abs{G^\geq \times G^\leq }
\end{equation}
holds, thereby proving the assertion. Notice that for every $f_h \in \mF_i$ such that $f_h(z)=1$ we have $\mI(z,h)=1$ for the corresponding $h\in \mH_i$; thus 
\begin{equation}
\label{eq:uioet}
-t_j\leq h(x_j)-y_j \leq t_j  
\Rightarrow
\begin{cases}
\ind_{ h(x_j) \leq y_j+t_j} = 1 \\
\ind_{ h(x_j) \geq y_j-t_j} = 1
\end{cases}
\Rightarrow
\begin{cases}
g_h^\leq (x_j,y_j+t_j)=1 \\
g_h^\geq (x_j,y_j-t_j)=1
\end{cases}. \footnote{Recall the definition of $g_h^\leq,g_h^\geq$ in Equation (\ref{eq:ggeqeqdef}).}
\end{equation}
Alternatively, if $f_h(z_j)=0$, we have $\mI(z_j,h)=0$ and
\begin{align}
\label{eq:uigfgfg}
&(h(x_j)-y_j < -t_j)  \lor ( h(x_j)-y_j > t_j) \Rightarrow 
(\ind_{ h(x_j) < y_j-t_j} = 1) \lor  (\ind_{ h(x_j)> y_j+t_j} =1) \nonumber \\
& \Rightarrow (\ind_{ h(x_j) \geq y_j-t_j} = 0) \lor  (\ind_{ h(x_j) \leq y_j+t_j} = 0) \Rightarrow (g_h^\geq (x_j,y_j-t_j)=0) \lor  (g_h^\leq (x_j,y_j+t_j) = 0).
\end{align}
By Equations (\ref{eq:uioet}) and (\ref{eq:uigfgfg}) we have
\begin{equation}
\label{eq:adfjnaf}
\begin{cases}
f_h(z_j)=1 \Rightarrow (g_h^\leq (x_j, y_j+t_j),g_h^\geq(x_j, y_j-t_j))=(1,1)\\
f_h(z_j)=0 \Rightarrow (g_h^\leq(x_j, y_j+t_j),g_h^\geq (x_j, y_j-t_j))\in \{(0,0),(0,1),(1,0)\}
\end{cases}.
\end{equation}
We define the mapping $M$ such that every vector $\left(\mI(z_1,h_1),\dots,\mI(z_m,h_1) \right)\in \mF_i \cap \mS$ is mapped to 
\[
\left(g_h^\geq (x_1, y_1-t_1),\dots, g_h^\geq (x_m, y_m-t_m),g_h^\leq (x_1, y_1+t_1),\dots ,g_h^\leq (x_m, y_m+t_m) \right)
\in G^\geq \times G^\leq.
\]
Namely, every vector obtained by applying $f_h$ on the sample $\mS$ is mapped to the vector formed by concatenating the two corresponding (same $h$) vectors from $G^\leq $ and $G^\geq $. Let $\bl b^1,\bl b^2 \in {\mF_i} \cap \mS$ such that  $b^1_j \neq b^2_j$ for at least one index $j\in[m]$, and w.l.o.g. let $b^1_j=1$ . Since $b^1_j=f_h(z_j)=\mI(z_j,h)$, Equation (\ref{eq:adfjnaf}) implies that $M(\bl b^1)_j = M(\bl b^1)_{j+m} = 1$, while at least one of $\{  M(\bl b^2)_j, M(\bl b^2)_{j+m}\}$ equals zero, thus $M(\bl b^1) \neq M(\bl b^2)$. Hence $M$ is an injection. 

Ultimately, notice that $\mS$ is arbitrary, thus
\[\Pi_{\mF_i}(m) =\max_{\mS\in \mZ^m} \abs{{\mF_i}\cap \mS }\leq \abs{G^\geq \times G^\leq } = \abs{G^\geq } \cdot \abs{ G^\leq } \leq \Pi_{\mG^{\geq}}(m) \cdot \Pi_{\mG^{\leq}}(m).\]

\end{proof}

\begin{claim}
\label{claim:binomialfactors}
$\left(\frac{n}{k}  \right)^k \leq {n \choose k} \leq \left(\frac{e n}{k}  \right)^k   $. 
\end{claim}
\begin{proof}[\textbf{Proof of Claim \ref{claim:binomialfactors}}] We prove the two claims separately.

$\bullet$ $\left(\frac{n}{k}  \right)^k \leq {n \choose k}$: fix $n$. We prove by induction for $k\leq n$. The assertion holds for $k=1$. For $k\geq 2$ and every $m$ such that $0<m<k\leq $ we have
\begin{equation*}
\label{eq:gkkmpdf}
k \leq n \Rightarrow \frac{m}{n} \leq \frac{m}{k} \Rightarrow 1 - \frac{m}{k} \leq 1- \frac{m}{n} \Rightarrow \frac{k-m}{k} \leq \frac{n-m}{n} \Rightarrow \frac{n}{k} \leq \frac{n-m}{k-m};
\end{equation*}
thus
\[
\left(\frac{n}{k}\right)^k = \frac{n}{k} \cdots \frac{n}{k}\leq \frac{n}{k}\frac{n-1}{k-1}\cdots\frac{n-k+1}{k-k+1} = \binom{n}{k}.
\]

$\bullet$ $ {n \choose k} \leq \left(\frac{e n}{k}  \right)^k  $:  since $e^k=\sum_{i=0}^\infty \frac{k^i}{i !}$ (the Taylor expansion of $e^k$), we have $e^k > \frac{k^k}{k!}$; thus, $\frac{1}{k!} < \left(\frac{e}{k}\right)^k$. As a result,
\[
\binom{n}{k}=\frac{n\cdot (n-1) \cdots (n-k+1}{k!} \leq \frac{n^k}{k!}<\left(\frac{en}{k}   \right)^k.
\]
\end{proof}

\begin{claim} 
\label{claim:flowerbound}
$\vc(\mF_i) \geq \pdim(\mH_i)$.
\end{claim}
\begin{proof}[\textbf{Proof of Claim \ref{claim:flowerbound}}]

Denote $\mS=\left(x_1,\dots,x_m\right)\in \mX^m$ and $\bl r \in \R^m$ such that $\mH_i$ pseudo-shatters $\mS$ with witness $\bl r$. We prove the claim by showing that we can construct $\mS' \in \mZ^m$ that is shattered by $\mF_i$. For every binary vector $\bl b\in\{-1,1\}^m$ there exists $h_{\bl b}$ such that $\sign(h_b(x_j)-r_j)=b_j$ for every $j\in [m]$. 

Denote $H = \{h_{\bl b} \in \mH_i \}_{\bl b \in \{-1,1\}^m}$ such that $\abs{H} =2^m$. For every $j$ such that $r_j \geq 0$, let
\[
y_j = \min \left\{0,\min_{h_{\bl b} \in H}  h_{\bl b}(x_j)    \right\}.
\]
In addition, for every $j$ such that $r_j < 0$, let
\[
y_j = \max \left\{0,\max_{h_{\bl b} \in H}  h_{\bl b}(x_j)    \right\},
\]
and denote
\[
\mS' = \left( x_j,y_j,\abs{r_j} -\sign(r_j) t_j  \right)_{j=1}^m.
\]
We now show that $\mF_i$ shatters $\mS'$. Fix an arbitrary $\bl b\in\{-1,1\}^m$, and observe that in case $r_j \geq 0$
\begin{equation}
\label{eq:asdasdasd}
\begin{cases}
b_j = 1 \\
b_j = -1
\end{cases}
\Rightarrow
\begin{cases}
h_{\bl b}(x_j) \geq r_j  \\
h_{\bl b}(x_j) < r_j
\end{cases}
\Rightarrow
\begin{cases}
h_{\bl b}(x_j) -y_j \geq r_j -y_j \\
h_{\bl b}(x_j) -y_j < r_j -y_j 
\end{cases}
\Rightarrow
\begin{cases}
\abs{h_{\bl b}(x_j) -y_j }\geq r_j -y_j \\
\abs{h_{\bl b}(x_j) -y_j }< r_j -y_j
\end{cases},
\end{equation}
where the last argument holds since $h_{\bl b}(x_j)\geq y_j$. Alternatively, if $r_j <0$ we have
\begin{equation}
\label{eq:asdasdyju}
\begin{cases}
b_j = 1 \\
b_j = -1
\end{cases}
\Rightarrow
\begin{cases}
h_{\bl b}(x_j) \geq r_j  \\
h_{\bl b}(x_j) < r_j
\end{cases}
\Rightarrow
\begin{cases}
-h_{\bl b}(x_j)+y_j \leq -r_j +y_j \\
-h_{\bl b}(x_j) +y_j> -r_j+y_j
\end{cases}
\Rightarrow
\begin{cases}
\abs{-h_{\bl b}(x_j)+y_j} \leq \abs{r_j} +y_j \\
\abs{-h_{\bl b}(x_j)+y_j} > \abs{r_j} +y_j 
\end{cases},
\end{equation}
where again the last set of inequalities holds since $h_{\bl b}(x_j)\leq y_j$. 
In case one of Equations (\ref{eq:asdasdasd}) and (\ref{eq:asdasdyju}) holds in equality, we can slightly shift $r_j$ (as was done in the proof of Claim \ref{claim:auxiliaryg}); hence we assume these are strict inequalities. The expression in Equation (\ref{eq:asdasdasd}) corresponds to $\mI(h_{\bl b},(x_j,y_j,r_j-y_j))$, while that of Equation (\ref{eq:asdasdyju}) corresponds to $\mI(h_{\bl b},(x_j,y_j,\abs{r_j}+y_j))$.

This analysis applies for every $\bl b$; hence, $\mF_i$ shatters $\mS'$ as required.
\end{proof}

\begin{claim}
\label{claim:rqep}
For a given $\bl h$ it holds that
\begin{align}
\label{eq:hoefbdvg}
\Pr_{\mS' \sim \mD^m}\left(\abs{\pi_i(\bl h)-\pi_i^{\mS'}(\bl h)}  \leq \frac{\epsilon}{2}   \right) \geq \frac{1}{2}.
\end{align}
\end{claim}
\begin{proof}[\textbf{Proof of Claim \ref{claim:rqep}}]
Recall Chebyshev's inequality
\[
\Pr \left(\abs{X - \E[x]} \geq \epsilon \right) \leq \frac{\Var(X)}{\epsilon^2}.
\]
Applying it for our problem, we get
\[
\Pr_{\mS' \sim \mD^m}\left(\abs{\pi_i(\bl h)-\pi_i^{\mS'}(\bl h)}  \geq \frac{\epsilon}{2}   \right) \leq \frac{\Var(\pi_i^{\mS'}(\bl h))}{\frac{\epsilon^2}{4}}.
\]
Notice that $\pi_i^{\mS'}(\bl h)$ is the average of independent random variables bounded in the $[0,1]$ segment; hence, by Popoviciu's inequality on variances we have 
\[
\Var(\pi_i^{\mS'}(\bl h)) \leq \frac{1}{4m}.
\]
Finally, for $m \geq \frac{2}{\epsilon^2 }$ it holds that
\[
\Pr_{\mS' \sim \mD^m}\left(\abs{\pi_i(\bl h)-\pi_i^{\mS'}(\bl h)}  \leq \frac{\epsilon}{2}   \right)=1-\Pr_{\mS' \sim \mD^m}\left(\abs{\pi_i(\bl h)-\pi_i^{\mS'}(\bl h)}  \geq \frac{\epsilon}{2}   \right) \geq 1-\frac{\frac{1}{4m}}{\frac{\epsilon^2}{4}}=1-\frac{1}{\epsilon^2m}\geq \frac{1}{2}.
\]
\end{proof}

\begin{claim} 
\label{claim:auxhoeffbin}
Let $m\geq 15$, $(X_i)_{i=1}^m$ be a sequence of i.i.d. Bernoulli r.v. with $p=\frac{1}{2}$, and let $\bar X = \frac{1}{m} \sum_{i=1}^m X_i$. Then $\Pr\left(\frac{1}{2} < \bar X < \frac{3}{4}\right) \geq \frac{1}{4}$ .
\end{claim}
\begin{proof}[\textbf{Proof of Claim \ref{claim:auxhoeffbin}}]
By Hoeffding's inequality we have $\Pr\left( \bar X \geq (p+\epsilon) \right)\leq \exp\left( -2\epsilon^2 m \right) $. Therefore
\begin{align*}
\Pr\left(\frac{1}{2} < \bar X < \frac{3}{4}\right) = \Pr\left( \bar X < \frac{3}{4}\right) - \Pr\left(\bar X \leq \frac{1}{2} \right) = 1 - \Pr\left( \bar X \geq \frac{3}{4}\right) -\frac{1}{2} \stackrel{\epsilon= \frac{1}{4}}{\geq} \frac{1}{2}-e^{\frac{-m}{8} } \stackrel{ m\geq 15}{\geq} \frac{1}{4}.
\end{align*}
\end{proof}

\section{Best response for the linear strategy space}
\label{sec:bestresponsregression}
In this section we build on the results of \cite{porat2017best} to devise a best response oracle for the linear hypothesis class. Let $\mH_i$ be the linear hypothesis class. Namely, $\mH_i$ is the function class such that each $\bl h_i \in \mH_i$ corresponds to the mapping  $\bl x \mapsto \bl h_i \cdot \bl x  $. \footnote{Note that now a strategy of player $i$ is itself a vector. To emphasize the vector arithmetics we use bold notation also for the instances, i.e. $\bl x$.  } Notice that under this representation, $\mH_i=\mathbb R^{n}$.

From here on we assume that the dimension of the input $n$ is fixed. By slightly modifying the algorithm given in \cite{porat2017best}, we show how player $i$ can compute a best response against any $\bl h_{-i}$. Indeed, when $n$ is fixed, our proposed algorithm is guaranteed to run in polynomial time. In particular, the run time complexity of our algorithm is independent of the hypothesis class complexity of other players.

As a first step, we consider the Partial Vector Feasibility problem (PVF). 

\begin{algorithm}[H]

\SetAlgorithmName{Problem}{PVF}{PVF}
\caption{\textsc{Partial Vector Feasibility (PVF)} }
\KwIn{a sequence of examples $\mS=(x_j,y_j,t_j)_{j=1}^m$, and a vector $\bl v\in \{1,a,b,0\}^m $ 	
}
\KwOut{a point $\bl h_i \in \R^n$ satisfying\\
	\begin{itemize}
	\item if $v_j=1$, then $\abs{ \bl h_i \cdot \bl x_j -y_j}<t_j$
	\item if $v_j=a$, then $\bl h_i \cdot \bl x_j-y_j>t_j$ \tcp{above}
	\item if $v_j=b$, then $\bl h_i \cdot \bl x_j  -y_j<-t_j$ \tcp{below}
	\item if $v_j=0$, there is no constraint for the $j$'th point
	\end{itemize}
	if such exists, and $\phi$ otherwise.
}
\label{problem:partialvectorfeasibility}
\end{algorithm}
Note that PVF is solvable in polynomial time via linear programming. We are now ready to present the Best Linear Response (BLR) algorithm. BLR has three main steps:
\begin{enumerate}
\item compute the potential payoff from each point in the sample
\item find all feasible subsets of points player $i$ can satisfy concurrently (i.e. vectors in $\mF_i \cap \mS$, where $\mF_i$ is as defined in Equation (\ref{eq:defoff}))
\item return a strategy that achieves the highest possible payoff.
\end{enumerate}

The first step consists of a straightforward computation. To motivate the second step, notice that if $\mI(z_j,\bl h_i)=0$, then either $\bl h_i \cdot \bl x_j-y_j>t_j$ or $\bl h_i \cdot \bl x_j  -y_j<-t_j$ holds. Therefore, we identify all vectors $\bl v=(v_1,\dots v_m)\in \{1,a,b\}^m$ such that there exists $\bl h_i \in \mH_i$ and $v_j = 1$ if $\mI(z_j,\bl h_i)=1$; $v_j=a$ if $\bl h_i \cdot \bl x_j-y_j>t_j$ (``above''); and $v_j = b$ if $\bl h_i \cdot \bl x_j  -y_j<-t_j$ (``below''). This is done by recursive partition of $\{1,a,b\}^m$, where in each iteration we consider only partial vectors, i.e. vectors with entries masked with ``0'' (see $\textsc{PVF}$). At the end of this step we have fully identified the set $\mF_i \cap \mS$, and have further information regarding zero entries of each vector in it - whether it corresponds to ``above'' or ``below'', in the aforementioned sense. Finally, we have all possible payoffs, so we pick a vector corresponding to the highest one. We then find a strategy that attains it by invoking \textsc{PVF} for the last time.

The above discussion is formulated via the following algorithm.

\RevertAlgoNumber
\begin{algorithm}[H]
\setcounter{algocf}{1}
\label{alg:rnwithoracle}
\DontPrintSemicolon
\KwIn{$\mS=(x_j,y_j,t_j)_{j=1}^m$, $\bl h_{-i}$}
\KwOut{A best response to $\bl h_{-i}$}
for every $j\in[m]$, $w_j \gets \frac{1}{\sum_{i'\neq i} \mI(z_j,h_{i'})+1} $ \label{alg:rnwithoracle:linew} \tcp*{player $i$ gets $w_j$ when satisfying $z_j$}
$\boldsymbol {v} \gets \{0 \}^m$  \tcp*{ $\boldsymbol {v} =(v_1,v_2,\dots,v_m)$}
$\mR_0 \gets \left\{\boldsymbol {v} \right\}$  \;
\For {$j=1$ to $m$} { \label{alg:rnwithoracle:lineifor}
	$\mR_j \gets \emptyset$ \;
	\For {$\boldsymbol {v} \in \mR_{j-1}$} { \label{alg:rnwithoracle:linefor}
		\For {$\alpha \in \{1,a,b\}$} { \label{alg:rnwithoracle:lineforallv}
			\If {$\textsc{PVF} \left(\mS,(\boldsymbol {v}_{-j},\alpha )\right) \neq \phi$}{ \label{alg:rnwithoracle:lineoracle}
			add $(\boldsymbol {v}_{-j},\alpha )$ to $\mR_j$  \tcp*{$(\boldsymbol {v}_{-j},\alpha )=(v_1,\dots v_{j-1},\alpha,v_{j+1},\dots,v_m)$} 
			}
		}
	}
}
$\boldsymbol {v}^* \gets  \argmax_{\boldsymbol {v} \in \mR_m} \sum_{j=1}^m w_j\ind_{v_j=1}$ \label{alg:rnwithoracle:argmax}
\tcp*{one such vector must exist}
\Return \textsc{PVF}$(\bl v^*)$ \; \label{alg:rnwithoracle:return}
\caption{\textsc{Best Linear Response} (BLR)}
\end{algorithm}

Theorem 2 in \cite{porat2017best} shows that the second step (for loop in line \ref{alg:rnwithoracle:lineifor}) is done in time $poly(m)$, and that $\mR_m$ is of polynomial size. The first step (line \ref{alg:rnwithoracle:linew}) and the last step (lines \ref{alg:rnwithoracle:argmax} and \ref{alg:rnwithoracle:return}) are clearly executed in polynomial time. Overall, \textsc{BLR} runs in polynomial time. In addition, since it considers all possible distinct strategies using $\mF_i \cap \mS$ and takes the one with the highest payoff, it indeed returns the best linear response with respect to $\bl h_{-i}$ in the empirical game.

}\fi} 

\end{document}